\documentclass[onecolumn,showpacs,preprint,amsmath,amssymb,superscriptaddress]{revtex4-1}

\usepackage{dcolumn}
\usepackage{bm}
\usepackage{hyperref}
\newcommand{\ignore}[1]{}
\usepackage{graphicx}
\usepackage{amsmath,amssymb,amsfonts,amsthm,color,latexsym}

\usepackage[utf8]{inputenc}
\usepackage[T1]{fontenc}
\usepackage{mathptmx}

\usepackage{soul}

\usepackage{tabstackengine}

\usepackage[utf8]{inputenc}
\usepackage[T1]{fontenc}
\usepackage{mathptmx}

\def\mb{\mathbf}

\def\bm{\boldsymbol}
\newcommand{\Rmnum}[1]{\uppercase\expandafter{\romannumeral #1\relax}}
\newcommand{\rmnum}[1]{\lowercase\expandafter{\romannumeral #1\relax}}
\def\mb{\mathbf}
\newcommand*\diff{\mathop{}\!\mathrm{d}}

\mathchardef\mhyphen="2D
\usepackage{amsmath}

\newtheorem{theorem}{Theorem}[section]

\newtheorem{proposition}[theorem]{Proposition}

\usepackage[normalem]{ulem}

\usepackage[export]{adjustbox}

\begin{document}
\preprint{}

\title{Construction of coarse-grained molecular dynamics with many-body non-Markovian memory}
\author{Liyao Lyu}
\affiliation{Department of Computational Mathematics, Science \& Engineering, Michigan State University, MI 48824, USA}%
\author{Huan Lei}
\email{leihuan@msu.edu}
\affiliation{Department of Computational Mathematics, Science \& Engineering, Michigan State University, MI 48824, USA}%
\affiliation{Department of Statistics \& Probability, Michigan State University, MI 48824, USA}%

\begin{abstract}
We introduce a machine-learning-based coarse-grained molecular dynamics (CGMD) model that
faithfully retains the many-body nature of the inter-molecular dissipative interactions.
Unlike the common empirical CG models, the present model is constructed based on the Mori-Zwanzig
formalism and naturally inherits the heterogeneous state-dependent memory term rather than matching
the mean-field metrics such as the velocity auto-correlation function. Numerical results show that
preserving the many-body nature of the memory term is crucial for predicting the collective transport and
diffusion processes, where empirical forms generally show limitations.
\end{abstract}

\maketitle

\section{Introduction}
Accurately predicting the collective behavior of multi-scale physical systems is a long-standing problem that requires the integrated modeling of the molecular-level interactions across multiple scales \cite{More_Anderson_Science_1972}. However, for systems without clear scale separation, there often exists no such a set of simple collective variables by which we can formulate the evolution in an analytic and self-determined way. One canonical example is coarse-grained molecular dynamics (CGMD). While the reduced degrees of freedom (DoFs) enable us to achieve a broader range of the spatio-temporal scale, the construction of truly reliable CG models remains highly non-trivial. A significant amount of work 
\cite{Torrie_Valleau_Umbrella_JCP_1977,Tuckerman_adiabatic_JCP_2002,Eric_TAMD_CPL_2006,Izvekov_Voth_JPC_2005,noid_multiscale_2008,Rudd_APS_1998,Lyubartsev_PRE_1995,Shell_JCP_2008,Kumar_Kollman_JCC_1992,Nielsen_JPCM_2004,Laio_Parrinello_PNAS_2002,Darve_Pohorille_JCP_2001}
(see also review 
\cite{noid_perspective_2013}), including recent machine learning (ML)-based approaches \cite{Stecher_Thomas_JCTC_2014, John_ST_JPCB_2017, Lemke_Tobias_JCTC_2017, PRL_DeePMD_2018, Zhang_DeePCG_JCP_2018}, have been devoted to constructing the conservative CG potential for retaining consistent static and thermodynamic properties.  However, accurate prediction of the CG dynamics further relies on faithfully modeling a memory term that represents the energy-dissipation processes arising from the unresolved DoFs; the governing equations generally become non-Markovian on the CG scale.  Moreover, such non-Markovian term often depends on the resolved variables in a complex way \cite{Satija_Makarov_JCP_2017,Luo_Xie_JPCB_2006,Best_Hummer_PNAS_2010,Plotkin_Wolynes_PRL_1998, Straus_Voth_JCP_1993, Morrone_Li_JPCB_2012,daldrop2017external} where the analytic formulation is generally unknown. Existing approaches often rely on empirical models such as Brownian motion \cite{Einstein_1905}, Langevin dynamics \cite{vankampen_book_2007}, and dissipative particle dynamics (DPD) \cite{Hoogerbrugge_SMH_1992, Espanol_SMO_1995}.  
Despite their broad applications, studies \cite{Lei_Cas_2010, hijon2010mori,Yoshimoto_Kinefuchi_PRE_2013} based on direct construction from full MD  show that the empirical (e.g., pairwise additive) forms can be insufficient to capture the state-dependent energy-dissipation processes due to the many-body and non-Markovian effects. Recent efforts \cite{lange2006collective, ceriotti2009langevin, baczewski2013numerical, Dav_Voth_JCP_2015, Lei_Li_PNAS_2016, Li_Darve_Kar_JCP_2017, russo2019deep, Jung_Hanke_JCTC_2017, Lee2019,ma2019coarse,ma2021transfer, Klippenstein_Vegt_JCP_2021,vroylandt2022likelihood, SheZ_JCP_2023, xie2023gle} model the memory term based on the generalized Langevin equation (GLE) and its variants (see also review \cite{klippenstein2021introducing}). While the velocity auto-correlation function (VACF) is often used as the target quantity for model parameterization, it 
is essentially a metric of the background dissipation under mean-field approximation. 
The homogeneous kernel overlooks the heterogeneity of the energy dissipation among the CG particles stemming from the many-body nature of the marginal probability density function of the CG variables. 
This limitation 
imposes a fundamental challenge for 
accurately modeling the local irreversible responses 
as well as the transport and diffusion processes on the collective scale.

This work aims to fill the gap with a new CG model that faithfully entails the state-dependent non-Markovian memory and the coherent noise. The model formulation can be loosely viewed as an extended dynamics of the CG variables joint with a set of non-Markovian features that embodies the many-body nature of the energy dissipation among the CG particles. Specifically, we treat each CG particle as an agent and seek a set of symmetry-preserving neural network (NN) representations that directly map its local environments 
to the non-Markovian friction interactions, and thereby circumvent the exhausting efforts of fitting the individual memory terms with a unified empirical form. Different from the ML-based potential model \cite{Zhang_DeePCG_JCP_2018}, the memory terms are represented by NNs in form of second-order tensors that strictly preserve the rotational symmetry and the positive-definite constraint.
Coherent noise can be introduced satisfying the second fluctuation-dissipation theorem and retaining consistent invariant distribution. 
Rather than matching the VACF, the model is trained based on the Mori-Zwanzig (MZ) projection formalism such that the effects of the unresolved interactions can be seamlessly inherited. We emphasize that the construction is not merely for mathematical rigor. Numerical results of a polymer molecule system show that the CG models with empirical memory forms are generally insufficient to capture heterogeneous inter-molecular dissipation that leads to inaccurate cross-correlation functions among the particles. Fortunately, the present model can reproduce both the auto- and cross-correlation functions. More importantly, it accurately predicts the challenging collective dynamics characterized by the hydrodynamic mode correlation and the van Hove function \cite{Van_Hove_Phys_Rev_1954} and shows the promise to predict the meso-scale transport and diffusion processes with molecular-level fidelity. 

\section{Methods}
Let us consider a full MD system consisting of $M$ molecules with a total number of $N$ atoms. The phase space vector is denoted by $\mb z = \left[\mb q, \mb p\right]$, where $\mb q, \mb p \in \mathbb{R}^{3N}$ represent the position and momentum vector, respectively. Given $\mb z(0) = \mb z_0$, the evolution follows $\mb z(t) = \rm{e}^{\mathcal{L}t} \mb z_0$, where $\mathcal{L}$ is the Liouville operator determined by the Hamiltonian $H(\mb z)$. The CG variables are defined by representing each molecule as a CG particle, i.e., 
$\phi(\mb z) = \left[\phi^Q (\mb z), \phi^P (\mb z)\right]$, where $\mb \phi^Q (\mb z) = \left[\mb Q_1, \mb Q_2, \cdots, \mb Q_M\right]$ and $\mb \phi^P (\mb z) = \left[\mb P_1, \mb P_2, \cdots, \mb P_M\right]$
represent the center of mass and the total momentum of individual molecules, respectively. $\mb Z(t)=[\mb Q(t), \mb P(t)]$ denote the map $\phi(\mb z(t))$ with $\mb z(0) = \mb z_0$. 
To construct the reduced model, we define the Zwanzig  projection operator as the conditional expectation with a fixed CG vector $\bm Z$, i.e., 
$\mathcal{P}_{\bm Z} f(\mb z) := \mathbb{E}[f(\mb z) \vert \phi(\mb z) = \bm Z ]$ under conditional density proportional to 
$\delta(\phi(\mb z) - \bm Z) \rm{e}^{-\beta H(\mb z)}$ and its orthogonal operator $\mathcal{Q}_{\bm Z} = \mb I - \mathcal{P}_{\bm Z}$.
Using Zwanzig's formalism \cite{Zwanzig1973}, the dynamics of $\mb Z(t)$ (see Appendix \ref{sec_SI:CG_dynamics}) can be written as 
\begin{equation}
\begin{split}
\dot{\mb Q} &= \mb{M}^{-1} \mb P \\
\dot{\mb P} &= -\nabla U(\mb Q) + \int_0^t \mb K(\mb Q(s), t-s) \mb V(s) \diff s + \mb R(t),
\end{split}
\label{eq:MZ_full}
\end{equation}
where $\mb M$ is the mass matrix and $\mb V = \mb{M}^{-1} \mb P$ is the velocity.  
$U(\mb Q)$ is the free energy under $\phi^{Q}(\mb z) \equiv \mb Q$. 
$\mb K(\mb Q, t) = \mathcal{P}_{\mb Z}[({\rm e}^{\mathcal{Q}_{\bm Z} \mathcal{L} t}\mathcal{Q}_{\bm Z} \mathcal{L}\mb P) (\mathcal{Q}_{\bm Z} \mathcal{L}\mb P)^T]$ is the memory representing the coupling between the CG and unresolved variables, and $\mb R(t)$ is the fluctuation force. 

Eq. \eqref{eq:MZ_full} provides the starting point to derive the various CG models. Direct evaluation of $\mb K(\mb Q, t)$ imposes a challenge as it relies on solving the full-dimensional orthogonal dynamics ${\rm e}^{\mathcal{Q}_{\bm Z} t}$. Further simplification $\mb K(\mb Q, t) \approx \theta(t)$ leads to the common GLE with a homogeneous kernel. Alternatively, the pairwise approximation $[\mb K(\mb Q, t)]_{ij} \approx \gamma (Q_{ij}) \delta(t)$ or $\gamma (Q_{ij})  \theta(t)$ 
leads to the standard DPD (M-DPD) and non-Markovian variants (NM-DPD), respectively. However, as shown below, such empirical forms are limited to capturing the state-dependence that turns out to be crucial for the dynamics on the collective scale, and motivates the present model retaining the many-body nature of $\mb K(\mb Q, t)$. 

To elaborate the essential idea, let us start with the Markovian approximation $\mb K(\mb Q, t) \approx -\bm \Gamma (\mb Q) \delta(t)$, where $\bm\Gamma(\mb Q) =  \bm\Xi(\mb Q) \bm\Xi(\mb Q)^T$ is the friction tensor preserving the semi-positive definite condition, 
and $\bm\Xi(\mb Q)$ needs to retain the translational, rotational, and permutational symmetry, i.e., 
\begin{equation}
\begin{split}
&\bm\Xi_{ij}(\mb Q_1 + \mb b, \cdots,\mb Q_M + \mb b) = \bm\Xi_{ij}(\mb Q_1, \cdots,\mb Q_M) \\ 
&\bm\Xi_{ij}(\mathcal{U}\mb Q_1, \cdots,\mathcal{U}\mb Q_M) = \mathcal{U}\bm\Xi_{ij}(\mb Q_1, \cdots,\mb Q_M) \mathcal{U}^T \\
&\bm\Xi_{\sigma(i)\sigma(j)}(\mb Q_{\sigma(1)}, \cdots,\mb Q_{\sigma(M)}) = \bm\Xi_{ij}(\mb Q_1, \cdots,\mb Q_M),
\end{split}
\label{eq:symmetry_constraint}
\end{equation}
where $\bm\Xi_{ij} \in \mathbb{R}^{3\times 3}$ represents the friction contribution of $j\mhyphen$th particle on $i\mhyphen$th particle, $\mb b\in \mathbb{R}^3$ is a translation vector, $\mathcal{U}$ is a unitary matrix, and $\sigma(\cdot)$ is a permutation function. 

To inherit the many-body interactions, we map the local environment of each CG particle into a set of generalized coordinates, 
i.e., $\hat{\mb Q}_i^k = \mb Q_i + \sum_{l\in\mathcal{N}_i} f^k(Q_{il}) \mb Q_{il}$, where $\mb f: \mathbb{R} \to \mathbb{R}^K$ is an encoder function to be learned, and $\mathcal{N}_i = \{l \vert Q_{il} < r_c\}$ is the neighboring index set of the $i\mhyphen$th particle within a cut-off distance $r_c$.  Accordingly, $\hat{\mb Q}_{ij} \in \mathbb{R}^{3\times K}$
represents a set of features that encode the inter-molecular configurations beyond the pairwise approximation.  The $k\mhyphen$th column ${\hat{\mb Q}_{ij}}^{k} = {\hat{\mb Q}_i}^{k} - {\hat{\mb Q}_j}^{k}$ preserves the translational and permutational invariance, by which we represent $\bm \Xi_{ij}$ by 
\begin{equation}
\begin{split}
\bm\Xi_{ij} = \sum_{k=1}^K h_k(\hat{\mb Q}_{ij}^T \hat{\mb Q}_{ij} ) \hat{\mb Q}_{ij}^k \otimes \hat{\mb Q}_{ij}^k + h_0(\hat{\mb Q}_{ij}^T \hat{\mb Q}_{ij} ) \mb I
\end{split}
\label{eq:Xi_markovian}
\end{equation}
where $h: \mathbb{R}^{K\times K}\to \mathbb{R}^{K+1}$ are encoder functions which will be represented by NNs. For $i=j$, we have $\bm\Xi_{ii} = -\sum_{j\in \mathcal{N}_i} \bm \Xi_{ij}$ based on the Newton's third law. We refer to Appendix \ref{sec_SI:NN} for the proof of the symmetry constraint \eqref{eq:symmetry_constraint}.

Eq. \eqref{eq:Xi_markovian} entails the state-dependency of the memory term $\mb K(\mb Z,t)$ under the Markovian approximation. To incorporate the non-Markovian effect, we embed the memory term within an extended Markovian dynamics \cite{ceriotti2009langevin} (see also Ref. \cite{SheZ_JCP_2023}). Specifically, we seek a set of non-Markovian features  $\bm \zeta := \left[\bm\zeta_{1}, \bm\zeta_{2}, \cdots, \bm\zeta_{n}\right]$, and construct the joint dynamics of $\left [\mb Z, \bm \zeta\right]$ by imposing the many-body form of the friction tensor between $\mb P$ and  $\bm \zeta$, i.e.,
\begin{equation}
\begin{split}
 \dot{\mb Q} &= \mb M^{-1}\mb P \\ 
\dot{\mb P} &= -\nabla U(\mb Q) + \bm \Xi(\mb Q)\bm\zeta \\
\dot{\bm\zeta} &= - \bm \Xi(\mb Q)^T \mb V - \bm\Lambda \bm\zeta + \bm \xi(t),
\end{split}
\label{eq:CGMD}
\end{equation}
where $\bm\Xi = \left[\bm \Xi^1 \bm \Xi^2 \cdots \bm \Xi^n\right]$ and each sub-matrix takes the form \eqref{eq:Xi_markovian} constructed by $\{\mb f^{i}(\cdot), \mb h^{i}(\cdot)\}_{i=1}^n$ respectively. $\bm\Lambda = \hat{\bm \Lambda}\otimes \mb I$ represents the coupling among $n$ features, where   $\mb I \in \mathbb{R}^{3N\times 3N}$ is the identity matrix and $\hat{\bm \Lambda} \in \mathbb{R}^{n\times n}$needs to satisfy the Lyapunov stability condition $\hat{\bm \Lambda} + \hat{\bm \Lambda}^T \ge 0$. 
Therefore, we write $\hat{\bm \Lambda} = \hat{\mb L}\hat{\mb L}^T + \hat{\mb L}^a$, where $\hat{\mb L}$ is a lower triangular matrix and $\hat{\mb L}^a$ is an anti-symmetry matrix which will be determined later.
By choosing the white noise $\bm\xi(t)$ following
\begin{equation}
\left\langle  \bm \xi(t) \bm \xi(t') \right\rangle = \beta^{-1}(\bm\Lambda + \bm\Lambda^T) \delta(t-t'),  
\end{equation}
we can show that the reduced model \eqref{eq:CGMD} retains the consistent invariant distribution, i.e.,  $\rho(\mb Q, \mb P, \bm\xi) \propto \exp[{-\beta(U(\mb Q) + \mb P^T \mb M^{-1} \mb P/2 + \bm \zeta^T \bm\zeta/2})]$ 
(see proof in Appendix \ref{sec_SI:invariant}).

Eq. \eqref{eq:CGMD} departs from the common CG models by retaining both the heterogeneity and non-Markovianity of the energy dissipation process. 
Rather than matching the mean-field metrics such as the homogeneous VACF, 
we learn the embedded memory 
$\bm\Xi(\mb Q(t)){\rm e}^{\bm\Lambda(t-s)}\bm\Xi(\mb Q(s))^T$ based on the MZ form. 
However, directly solving the orthogonal dynamics ${\rm e}^{\mathcal{Q}_{\bm Z}Lt}$ is computationally intractable. Alternatively, we introduce the constrained dynamics $\tilde{\mb z}(t)= {\rm e}^{\mathcal{R}t} \mb z(0)$ following Ref. \cite{hijon2010mori}. Based on the observation $\mathcal{P}\mathcal{Q} = \mathcal{P}\mathcal{R}\equiv 0$, we sample the MZ form from $\tilde{\mb z}(t)$, i.e., $\mb K_{MZ}(\bm Z, t) = \mathcal{P}_{\bm Z}[({\rm e}^{\mathcal{R} t}\mathcal{Q}_{\bm Z} \mathcal{L}\mb P) (\mathcal{Q}_{\bm Z} \mathcal{L}\mb P)^T]$ and the memory of the CG model reduces to $\mb K_{CG}(\bm Z, t) = \bm\Xi(\mb Q){\rm e}^{\bm\Lambda t}\bm\Xi(\mb Q)^T$. 
This enables us to train the CG models in terms of the encoders  $\{\mb f^{i}(\cdot), \mb h^{i}(\cdot)\}_{i=1}^n$ and matrices $\hat{\mb L}$ and $\hat{\mb L}^a$ by minimizing the empirical loss 
\begin{equation}
L = \sum_{l=1}^{N_s} \sum_{j=1}^{N_t} \left\Vert  \mb K_{CG}(\mb Z^{(l)}, t_j) -  \mb K_{MZ}(\mb Z^{(l)}, t_j)\right\Vert^2,  
\end{equation}
where $l$ represents the different CG configurations (see Appendix \ref{sec_SI:training} for details in training). 

\section{Numerical results}
To demonstrate the accuracy of the present model, we consider a full
micro-scale model of a star-shaped polymer melt system similar to Ref. \cite{hijon2010mori}, where each molecule consists of $73$ atoms. The atomistic interactions are modeled by the Weeks-Chandler-Anderse potential and the Hookean bond potential.
The full system consists of $486$ molecules in a cubic domain $90\times90\times90$ with periodic boundary conditions. The Nos\'{e}-Hoover thermostat \cite{Nose_Mol_Phys_1984,Hoover1985} is employed to equilibrate the system with $k_BT =  4.0$ and micro-canonical ensemble simulation is conducted during the production stage (see Appendix \ref{sec_SI:MD}) for details). Below we compare different dynamic properties predicted by the full MD and the various CG models.  For fair comparisons, 
we use the same CG potential $U(\mb Q)$ constructed by the DeePCG scheme \cite{Zhang_DeePCG_JCP_2018} for all the CG models; the differences in dynamic properties solely arise from the different formulations of the memory term.

\begin{figure}[htp]
    \centering
    \includegraphics[width=0.95\linewidth]{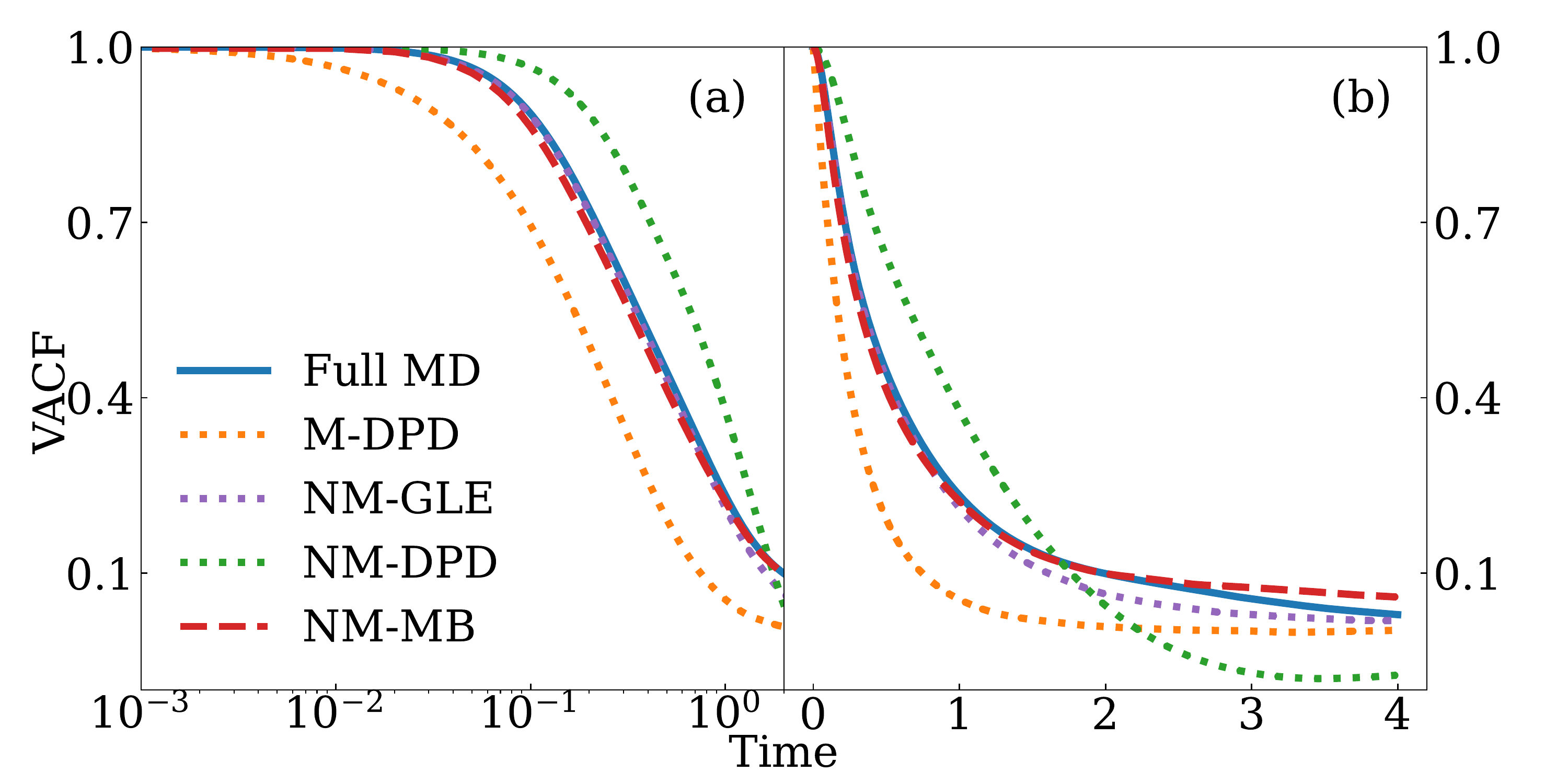}
    \caption{The VACF of the full MD and CG models with various memory formulations in (a) semi-log scale (b) original scale. ``M'' and ``NM'' represent Markovian and Non-Markovian; GLE, DPD, and MB represent state-independent, pairwise, and the present (NM-MB) model retaining the many-body  effects, respectively.}
    \label{fig:vacf}
\end{figure}
Let us start with the VACF which has been broadly used in CG model parameterization and validation. As shown in Fig. \ref{fig:vacf}, the predictions from the present model (NM-MB) show good agreement with the full MD results.
In contrast, the CG model with the memory term represented by the pairwise decomposition and Markovian approximation (i.e., the standard M-DPD form) yields apparent deviations. The form of the pairwise decomposition with non-Markovian approximation (NM-DPD) shows improvement at a short time scale but exhibits large deviations at an intermediate scale. Such limitations indicate pronounced many-body effects in the energy dissipation among the CG particles. Alternatively, if we set the VACF as the target quantity, we can parameterize the empirical model such as GLE by matching the VACF predicted by the full MD. Indeed, the prediction from the constructed GLE recovers the MD results. However, as shown below, this form over-simplifies the heterogeneity of the memory term and leads to inaccurate predictions on the collective scales.

\begin{figure}[htpb]
    \centering
    \includegraphics[width=0.95\linewidth]{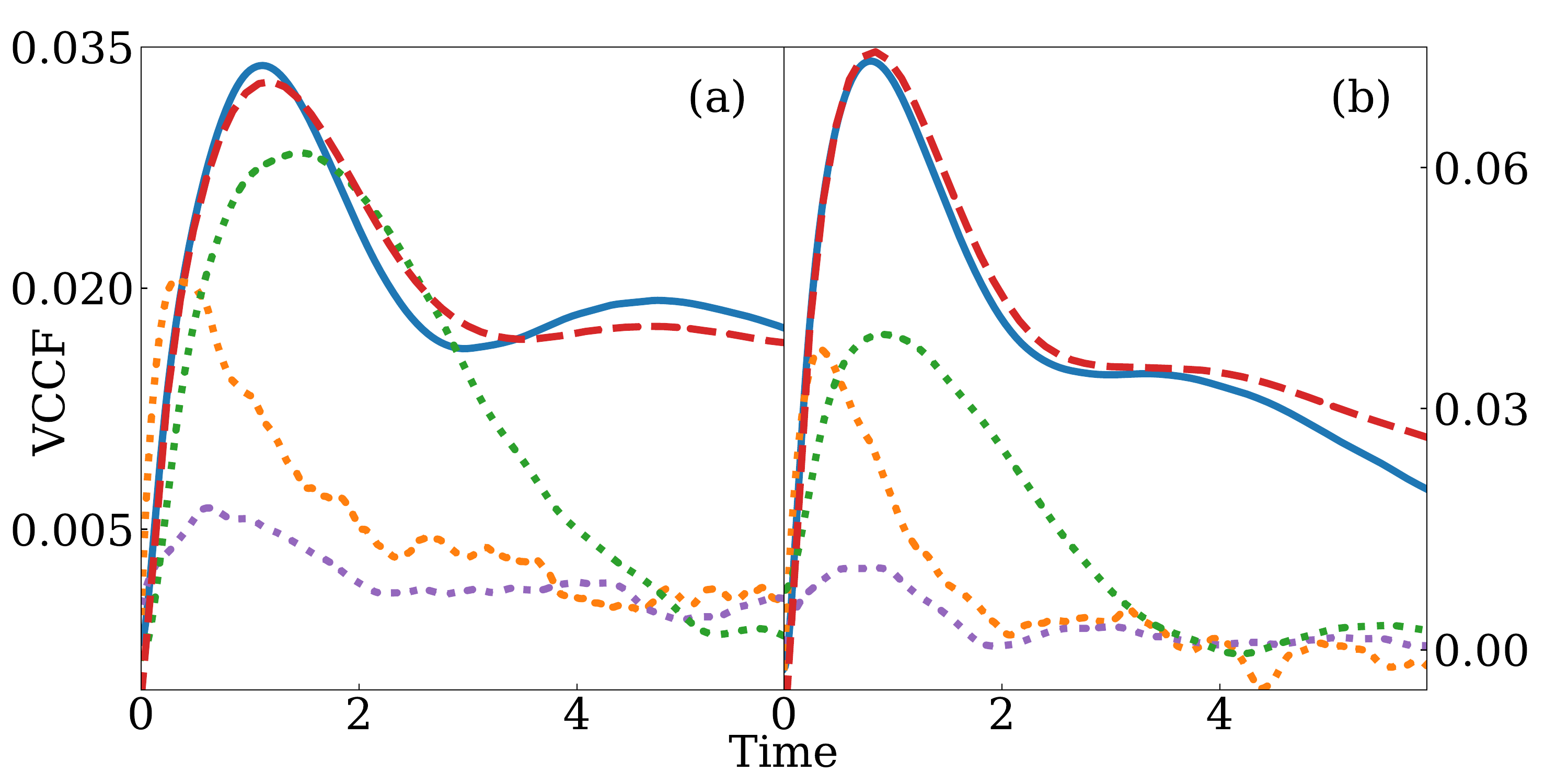}
    \caption{The VCCF $C^{xx}(t; r_0)$ predicted by the full MD and different CG models with initial distance (a) $10<r_0<11$ and (b) $14<r_0<15$. Same line legend as Fig. \ref{fig:vacf}.}
    \label{fig:vccf}
\end{figure}

Fig. \ref{fig:vccf} 
shows the velocity cross-correlation function (VCCF) between two CG particles, i.e., $C^{xx}(t; r_0) = \mathbb{E}[\mb V_i(0)\cdot\mb V_j(t) \vert Q_{ij}(0) = r_0]$, where $r_0$ represents the initial distance. Similar to VACF, the present model (NM-MB) yields good agreement with the full MD results. However, the predictions from other empirical models, including the GLE form, show apparent deviations. Such limitations arise from the inconsistent representation of the local energy dissipation and 
can be understood as following. The VACF represents the energy dissipation on each particle as a homogeneous background heat bath; it is essentially a mean-field metric and can not characterize the dissipative interactions among the particles. Hence, the reduced models that only recover the VACF could be insufficient to retain the consistent local momentum transport and the correlations among the particles.
 
\begin{figure}[htpb]
    \centering
    \includegraphics[width=0.95\linewidth]{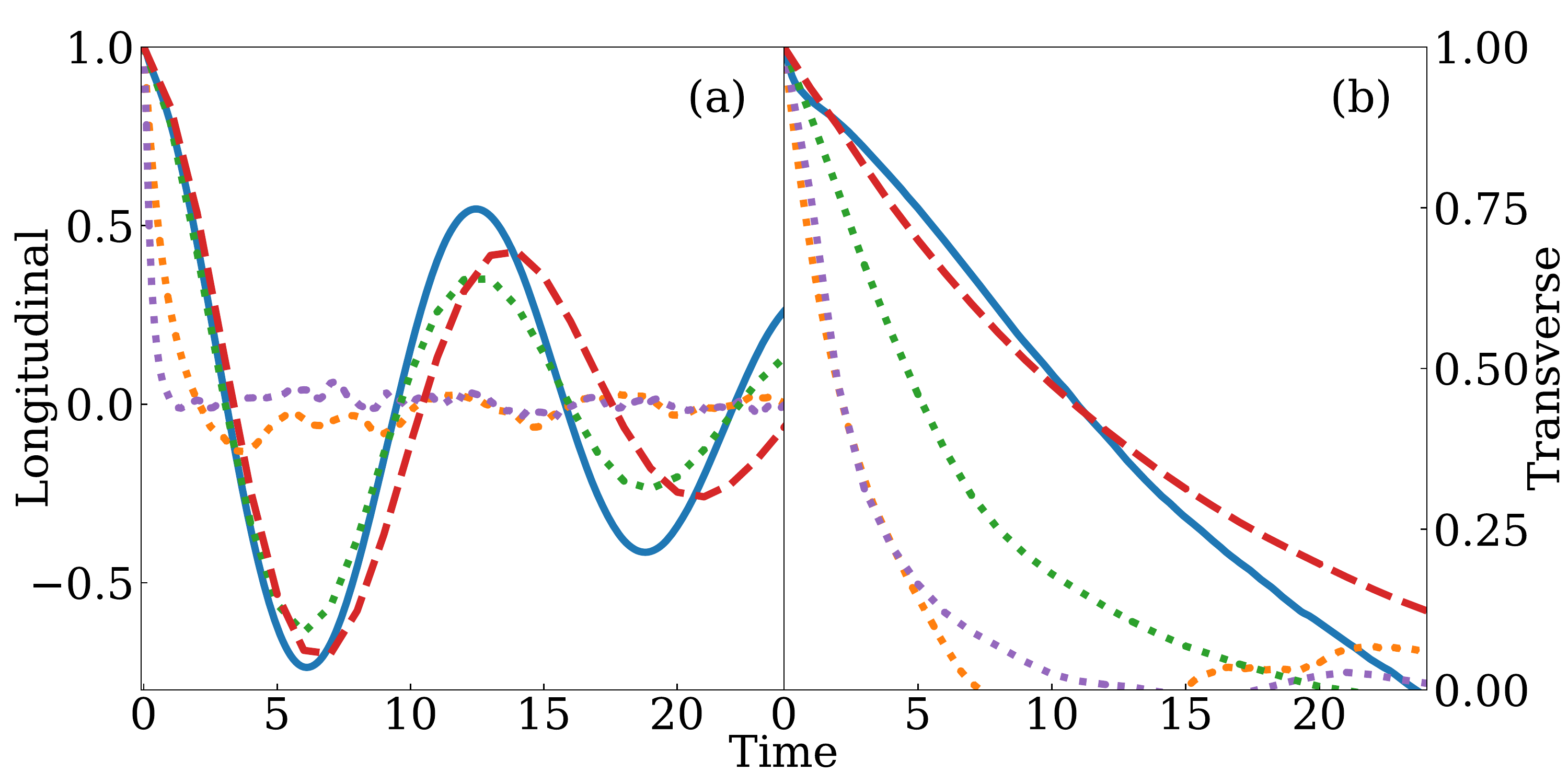}
    \caption{(a) Longitudinal and (b) Transverse hydrodynamic modes predicted by MD and different CG models. Same line legend as Fig. \ref{fig:vacf}.}
    \label{fig:hydromode}
\end{figure} 
Furthermore, the various empirical models for local energy dissipations can lead to fundamentally different transport processes on the collective scale. Fig. \ref{fig:hydromode} 
 shows the normalized correlations of the longitudinal and transverse hydrodynamic modes \cite{Theory_simple_liquids_book_2007}, i.e., $C_L(t) = \langle \tilde{u}_1(t)\tilde{u}_1(0)\rangle$ and $C_T(t) = \langle \tilde{u}_2(t)\tilde{u}_2(0)\rangle$, where $\tilde{\mb u} = 1/M\sum_{j=1}^M \mb V_j {\rm e}^{i\mb k \cdot \mb Q_j}$, $\mb k$ is the wave vector, and the subscripts $1$ and $2$ represent the direction parallel and perpendicular to $\mb k$, respectively. 
Similar to the VCCF, the prediction from the present model (NM-MB) agrees well with the MD results while other models show apparent deviations. In particular, the prediction from the GLE model shows strong over-damping due to the ignorance of the inter-molecule dissipations.

\begin{figure}[htp]
    \centering
    \includegraphics[width=0.95\linewidth]{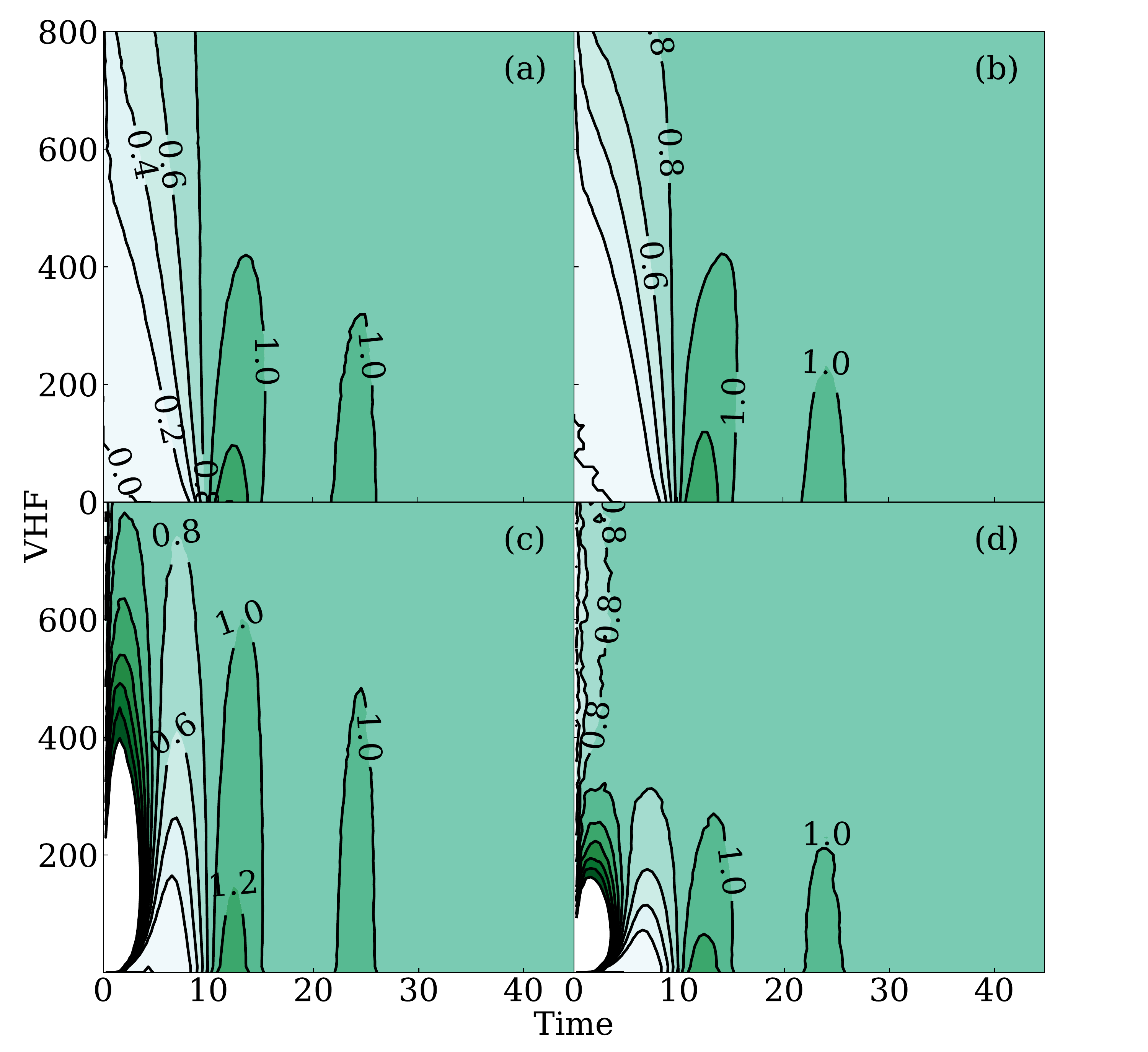}
    \caption{The van Hove function predicted by (a) full MD (b) the present NM-MB model (c) NM-DPD model (d) GLE model.}
    \label{fig:vhf}
\end{figure}
Finally, we examine the diffusion process on the collective scale. Fig. \ref{fig:vhf} shows the van Hove function that characterizes the evolution of the inter-particle structural correlation defined by $G(r, t) \propto \frac{1}{M^2} \sum_{j\neq i}^M \delta(\Vert \mb Q_i(t) - \mb Q_j(0)\Vert - r)$. At $t=0$, $G(r, t)$ reduces to the standard radial distribution function where all the CG models can recover such initial conditions. However, for $t > 0$, predictions from the models with the pairwise decomposition (NM-DPD) and the GLE form show apparent deviations. Specifically, at an early stage near $t = 50$, the neighboring particles begin to artificially jump into the region near the reference particle, violating the fluid-structure thereafter. In contrast, the present model (NM-MB) shows consistent predictions of the structure evolution over a long period until $t= 1000$, when the initial fluid structure ultimately diffuses into a homogeneous state. 

\section{Summary}
To conclude, we developed a CG model that faithfully accounts for the broadly overlooked many-body nature of the non-Markovian memory term. We show that retaining the heterogeneity and the strong correlation of the local energy dissipation is crucial for accurately predicting the cross-correlation among the CG particles, which, however, can not be fully characterized by the mean-field metrics such as VACF. More importantly, the memory form representing the inter-molecule energy dissipations may play a profound role in the transport and diffusion processes on the collective scale. In particular, the present model accurately predicts the hydrodynamic mode correlation and the van Hove function where empirical forms show limitations,
and therefore, shows the promise to study challenging problems relevant to the meso-scale transition and synthesis processes. 

\begin{acknowledgements}
The work is supported in part by the National Science Foundation under Grant DMS-2110981 and 
the ACCESS program through allocation MTH210005.
\end{acknowledgements}

\appendix
\section{Dynamics of the coarse-grained variables}
\label{sec_SI:CG_dynamics}
We consider a full MD system consisting of $M$ molecules with a total number of $N$ atoms. The phase space vector is denoted by $\mb z = \left[\mb q, \mb p\right]$, where $\mb q\in \mathbb{R}^{3N}$ and $\mb p \in \mathbb{R}^{3N}$ represent the position and momentum vector, respectively. The coarse-grained (CG) variables are defined by representing each molecule as a CG particle, i.e., $\phi(\mb z) = \left[\bm\phi^{Q}(\mb q), \bm\phi^{P}(\mb q)\right]$, where $\bm\phi^{Q} = \left[\mb Q_1, \mb Q_2, \cdots \mb Q_M\right]$ and $\bm\phi^{P} = \left[\mb P_1, \mb P_2, \cdots \mb P_M\right]$ represent the center of mass (COM) and the total momentum of the individual molecules. Let $\mb Z(t) = [\mb Q(t), \mb P(t)]$ denote the map $\phi(\mb z(t))$ with $\mb z(0) = \mb z_0$. 
Using the Koopman operator \cite{Koopman315}, $\mb Z(t)$ can be mapped from the initial values, i.e., 
\begin{equation}
\mb Z(t) = {\rm e}^{\mathcal{L} t}\mb Z(0),
\label{eq:Z_koopman}
\end{equation}
where $\mathcal{L}$ is the Liouville operator determined by the full-model Hamiltonian $H(\mb z)$. Below we derive the reduced model by choosing CG variables $\mb Z$ as a linear mapping of the full phase-space 
vector $\mb z$ (see also Ref. \cite{KinHyo07}) and we refer to Refs. \cite{hijon2010mori, Darve_PNAS_2009} for discussions of the more general cases.  

Following Zwanzig's approach, we define a projection operator as the conditional expectation with a fixed CG vector $\bm Z$, i.e., $\mathcal{P}_{\bm Z} f(\mb z) := \int \delta(\bm\phi(\mb z) - \bm Z) \rho_0(\mb z) f(\mb z) \diff \mb z / \Omega(\bm Z)$, where $\rho_0(\mb z) \propto \rm{e}^{-\beta H(\mb z)}$ represents the equilibrium density function and 
$\Omega(\bm Z) =  \int  \delta(\bm\phi(\mb z) - \bm Z) \rho_0(\mb z) \diff \mb z$.
Also, we define an orthogonal operator $\mathcal{Q}_{\bm Z} = \mb I - \mathcal{P}_{\bm Z}$. Using Eq. \eqref{eq:Z_koopman}, we have $\dot{\mb Z}(t) =  {\rm e}^{\mathcal{L} t}\mathcal{P}_{\bm Z} \mathcal{L} \mb Z(0) +  {\rm e}^{\mathcal{L} t} \mathcal{Q}_{\bm Z} \mathcal{L} \mb Z(0)$. In particular, we choose $\bm Z = \mb Z(0)$. Using the Duhamel-Dyson identity, we can write the dynamics of $\mb Z(t)$ as
\begin{equation}
\dot{\mb Z}(t) = {\rm e}^{\mathcal{L}t}\mathcal{P}_{\bm Z}\mathcal{L}\mb Z(0) + \int_0^t \diff s {\rm e}^{\mathcal{L}(t-s)}\mathcal{P}_{\bm Z} \mathcal{L} {\rm e}^{\mathcal{Q}_{\bm Z}\mathcal{L}s} \mathcal{Q}_{\bm Z}\mathcal{L}\mb Z(0)   + {\rm e}^{\mathcal{Q}_{\bm Z}\mathcal{L}t} \mathcal{Q}_{\bm Z}\mathcal{L}\mb Z(0). 
\label{eq:Z_MZ_full}
\end{equation}

Let us start with the mean-field term $\mathcal{P}_{\bm Z}\mathcal{L}\mb Z(0)$. For the present study, the CG variables are linear functions of $\mb z$. Therefore, we have $\mathcal{P}_{\bm Z} \mathcal{L} \mb Q = \mathcal{L} \mb Q = \mb M^{-1} \mb P$, i.e., $\mathcal{Q}_{\bm Z} \mathcal{L} \mb Q \equiv 0$. For $\mathcal{P}_{\bm Z} \mathcal{L} \mb P$ associated with the $i\mhyphen$th CG particle, we have
\begin{equation}
\begin{split}
\mathcal{P}_{\bm Z} \mathcal{L} \mb P_i &= \int  \delta(\bm\phi(\mb z) - \bm Z) \rho_0(\mb z) \mathcal{L} \mb P_i \diff \mb z / \Omega(\bm Z)  \\
&= \int  \delta(\bm\phi(\mb z) - \bm Z) \rho_0(\mb z) (-\sum_{i\in \mathcal{N}_i} \nabla_{\mb q_i} H(\mb z)) \diff \mb z / \Omega(\bm Z) \\ 
&= \int  \delta(\bm\phi(\mb z) - \bm Z) (\beta^{-1} \sum_{i\in \mathcal{N}_i} \nabla_{\mb q_i}) \rho_0(\mb z)  \diff \mb z / \Omega(\bm Z) \\
&= \beta^{-1} \nabla_{\bm Q_i} \int  \delta (\bm\phi^{Q}(\mb q) - \bm Q) \rho_0(\mb q)  \diff \mb q / \int  \delta (\bm\phi^{Q}(\mb q) - \bm Q)  \rho_0(\mb q) \diff \mb q \\
&= -\nabla_{\bm Q_i} U(\bm Q),
\end{split}
\label{eq:free_energy_derivation}
\end{equation}
where $\mathcal{N}_i$ represents the index set of the atoms that belongs to the $i\mhyphen$th molecule, and $U(\bm Q)$ represents the free energy defined by $U(\bm Q) = -\beta^{-1} \ln \left [\int  \delta (\bm\phi^{Q}(\mb q) - \bm Q)  \rho_0(\mb q) \diff \mb q \right]$. 

For the memory term $\mathcal{P}_{\bm Z} \mathcal{L} {\rm e}^{\mathcal{Q}_{\bm Z}\mathcal{L}s} \mathcal{Q}_{\bm Z}\mathcal{L}\mb P$ associated with the $i\mhyphen$th CG particle, we have
\begin{equation}
\begin{split}
\mathcal{P}_{\bm Z} \mathcal{L} {\rm e}^{\mathcal{Q}_{\bm Z}\mathcal{L}s} \mathcal{Q}_{\bm Z}\mathcal{L}\mb P_i &= \int \rho_0(\mb z) \delta(\bm\phi(\mb z) - \bm Z)   \mathcal{L} {\rm e}^{\mathcal{Q}_{\bm Z}\mathcal{L}s} \mathcal{Q}_{\bm Z}\mathcal{L}\mb P_i   \diff \mb z / \Omega(\bm Z) \\
&=  \int \rho_0(\mb z) (\mathcal{L} \phi(\mb z) \cdot \nabla_{\bm Z}) \delta(\bm\phi(\mb z) - \bm Z)   {\rm e}^{\mathcal{Q}_{\bm Z}\mathcal{L}s} \mathcal{Q}_{\bm Z}\mathcal{L}\mb P_i   \diff \mb z / \Omega(\bm Z) \\
&=  \int \rho_0(\mb z) (\mathcal{Q}_{\bm Z} \mathcal{L} \mb P \cdot \nabla_{\bm P}) \delta(\bm\phi(\mb z) - \bm Z)   {\rm e}^{\mathcal{Q}_{\bm Z}\mathcal{L}s} \mathcal{Q}_{\bm Z}\mathcal{L}\mb P_i   \diff \mb z / \Omega(\bm Z) ~({\rm by}~ \mathcal{Q}_{\bm Z}\mathcal{L}\mb Q\equiv 0)\\
&=  \nabla_{\bm P} \cdot \int \rho_0(\mb z) \delta(\bm\phi(\mb z)   - \bm Z) (\mathcal{Q}_{\bm Z}\mathcal{L} \mb P ) \otimes {\rm e}^{\mathcal{Q}_{\bm Z}\mathcal{L}s} \mathcal{Q}_{\bm Z}\mathcal{L}\mb P_i   \diff \mb z / \Omega(\bm Z)\\ 
&=  \nabla_{\bm P} \cdot \underbrace{\left( \int \rho_0(\mb z) \delta(\bm\phi(\mb z)   - \bm Z) (\mathcal{Q}_{\bm Z} \mathcal{L} \mb P ) \otimes {\rm e}^{\mathcal{Q}_{\bm Z}\mathcal{L}s} \mathcal{Q}_{\bm Z}\mathcal{L}\mb P_i   \diff \mb z / \Omega(\bm Z) \right)}_{\tilde{\mb K}_{i,}(\bm Z, s)}  \\
&- \tilde{\mb K}_{i,}(\bm Z, s) \cdot \nabla_{\bm P} \left(1/\Omega(\bm Z)\right) \Omega(\bm Z). 
\end{split}
\label{eq:memory_derivation}
\end{equation}
Furthermore, we take the assumption that the memory kernel only depends on the positions of the CG particles $\bm Q$, i.e., $\nabla_{\bm P} \cdot \tilde{\mb K}(\bm Z, s) \equiv 0$. Also, similar to the derivation in Eq. \eqref{eq:free_energy_derivation}, we note that 
\begin{equation}
\Omega(\bm Z) \propto   \int  \delta (\bm\phi^{Q}(\mb q) - \bm Q) \rho_0(\mb q)  
\delta (\bm\phi^{P}(\mb q) - \bm P) e^{-\beta \mb P^T \mb M^{-1} \mb P/2} \diff \mb z
\propto e^{-\beta \bm P^T \mb M^{-1} \bm P/2}.
\end{equation}
Therefore, Eq. \eqref{eq:memory_derivation} can be further simplified as
\begin{equation}
\mathcal{P}_{\bm Z} \mathcal{L} {\rm e}^{\mathcal{Q}_{\bm Z}\mathcal{L}s} \mathcal{Q}_{\bm Z}\mathcal{L}\mb P_i =  -\beta \tilde{\mb K}_{i,}(\bm Q, s) \cdot \mb M^{-1} \bm P. 
\label{eq:memory_derivation_simple}
\end{equation}

With Eqs. \eqref{eq:free_energy_derivation} \eqref{eq:memory_derivation_simple}, we can show that the dynamics of $\mb Z = [\mb Q, \mb P]$ can be written as 
\begin{equation}
\begin{split}
\dot{\mb Q} &= \mb M^{-1} \mb P \\ 
\dot{\mb P} &= - \nabla U(\mb Q) - \int_0^t \mb K(\mb Q(t-s), s) \mb V(t-s) \diff s + \mb R(t),
\end{split}    
\end{equation}
where $\mb K(\mb Q, s) = \beta \tilde{\mb K}(\mb Q, s)$ and $\mb R(t) = {\rm e}^{\mathcal{Q}_{\bm Z}\mathcal{L}t} \mathcal{Q}_{\bm Z}\mathcal{L}\mb Z(0)$ is modeled as a random process representing the different initial condition $\mb z_0$ with $\phi(\mb z_0) = \bm Z$.

\section{The micro-scale model of the polymer melt system}
\label{sec_SI:MD}
We consider the micro-scale model of a star-shaped polymer melt system 
similar to Ref. \cite{hijon2010mori}. Each polymer molecule consists of a ``center'' atom connected by
$12$ arms with $6$ atoms per arm. The potential function is governed by the pairwise and bond interactions, i.e.,
\begin{equation}
V(\mb q) = \sum_{i\neq j} V_{p}(q_{ij}) + \sum_k V_{b}(l_k),
\label{eq:MD_potential}
\end{equation}
where $V_{p}$ is the pairwise interaction between both the intra- and inter-molecular atoms except the bonded pairs. $q_{ij} = \Vert \mb q_i - \mb q_j\Vert$ 
is the distance between the $i\mhyphen$th and $j\mhyphen$th atoms. $V_{p}$ takes the form of the Lennard–Jones potential with cut-off $r_c$, i.e., 
\begin{equation}
V_{p}(r) = \begin{cases} V_{\rm LJ}(r) - V_{\rm LJ}(r_c), ~r < r_c  \\
0, ~r \ge r_c 
\end{cases} \quad
\quad V_{\rm LJ}(r) = 4\epsilon \left[\left(\frac{\sigma}{r}\right)^{12} -  \left(\frac{\sigma}{r}\right)^6\right],
\label{eq:LJpotential}
\end{equation}
where $\epsilon = 1.0$ is the dispersion energy and $\sigma = 2.415$ is the hardcore distance. Also we choose $r_c = 2^{1/6}\sigma$ so that $V_p$ recovers the Weeks-Chandler-Andersen potential. 
$V_{b}$ is the bond interaction between the neighboring  particles of each polymer arm and $l_k$ is the length of the $k\mhyphen$th bond. 
The bond potential $V_b$ is chosen to be the harmonic potential, i.e.,
\begin{equation}
V_{b}(l) = \frac{1}{2}k_s (l - l_0)^2, 
\label{eq:harmonic}
\end{equation}
where $k_s = 1.714$ and $l_0 = 1.615$ represent the elastic coefficient and the equilibrium length $l_0$, respectively. The atom mass is chosen to be unity. The full system consists of $N = 486$ polymer molecules in a cubic domain $90\times90\times90$ with periodic boundary condition imposed along each direction. The Nos\'{e}-Hoover thermostat is employed to equilibrate the system with $k_BT =  4.0$ and micro-canonical ensemble simulation is conducted during the production stage.

\section{Invariant density function of the CG model}
\label{sec_SI:invariant}
The reduced model takes the following form 
\begin{equation}
\begin{split}
\dot{\mb Q} &= \mb M^{-1}\mb P \\
\dot{\mb P} &= -\nabla U(\mb Q) + \bm \Xi(\mb Q)\bm\zeta \\
\dot{\bm\zeta} &= - \bm \Xi(\mb Q)^T \mb V - \bm\Lambda \bm\zeta + \bm \xi(t), 
\end{split}
\label{eq:CGMD_app}
\end{equation}
where $\bm\Xi = \left[\bm \Xi^1 \bm \Xi^2 \cdots \bm \Xi^n\right]$ represents a set of non-Markovian features. It resembles the extended dynamics for the GLE proposed in Ref. \cite{SheZ_JCP_2023} except that the coupling between $\mb P$ and the features $\bm \zeta$ are represented by the state-dependent friction tensor $\bm \Xi(\mb Q)$ retaining the many-body nature. By properly choosing the white noise $\bm\xi(t)$, we can show that model \eqref{eq:CGMD_app} retains the invariant density function consistent with the full MD model. 
\begin{proposition} 
By choosing the white noise $\bm\xi(t)$ following
\begin{equation}
\left\langle  \bm \xi(t) \bm \xi(t') \right\rangle = \beta^{-1}(\bm\Lambda + \bm\Lambda^T) \delta(t-t'),
\end{equation}
Model \eqref{eq:CGMD_app} retains the consistent invariant distribution 
\begin{equation}
\rho_{\rm eq}(\mb Q, \mb P, \bm\xi) \propto \exp[{-\beta(U(\mb Q) + \mb P^T \mb M^{-1} \mb P/2 + \bm \zeta^T \bm\zeta/2})]
\label{eq:rho_eq_app}
\end{equation}
\end{proposition} 

\begin{proof}
Let $\tilde{\mb Z} = [\mb Q, \mb P, \bm\zeta]$ denote the resolved variables and $W(\tilde{\mb Z}) = U(\mb Q) + \mb P^T \mb M^{-1} \mb P/2 + \bm \zeta^T \bm\zeta/2$ the free energy of the extended dynamics. Model \eqref{eq:CGMD_app} can be written as the following gradient dynamics
\begin{equation*}
\frac{\diff \tilde{\mb Z}}{\diff t} = 
\underbrace{\begin{pmatrix} 0 &\mb I &0 \\
-\mb I& 0 & \bm\Xi(\mb Q) \\
 0 &\bm\Xi(\mb Q)^T &\bm\Lambda
\end{pmatrix}}_{\mb G(\mb Q)}\nabla_{\tilde{\mb Z}} W(\tilde{\mb Z}) + \tilde{\bm\xi}(t),
\end{equation*}
where $\tilde{\bm\xi}(t) = [0, 0, \bm\xi(t)]$. Accordingly, the Fokker-Planck equation takes the form 
\begin{equation*}
\frac{\partial \rho(\tilde{\mb Z}, t)}{\partial t} = \nabla \cdot \left(-\mb G(\mb Q) \nabla W(\tilde{\mb Z}) \rho(\tilde{\mb Z}, t) - \frac{1}{2} \beta^{-1}(\mb G(\mb Q) + \mb G(\mb Q)^T) \nabla \rho(\tilde{\mb Z}, t) \right). 
\end{equation*}
Plug Eq. \eqref{eq:rho_eq_app} into the above equation, we have
\begin{equation}
\begin{split}
\nabla \cdot \left(\beta^{-1} \mb G(\mb Q) \nabla \rho_{\rm eq}(\mb z, t)   - \frac{1}{2} \beta^{-1}(\mb G(\mb Q)+ \mb G(\mb Q)^T) \nabla \rho_{\rm eq}(\mb z, t) \right) 
&=  \beta^{-1} \nabla \cdot  \left(\tilde{\bm\Lambda}^A   \nabla \rho_{\rm eq}(\mb z, t)\right) \\
&\equiv 0,
\end{split}
\end{equation}
where $\tilde{\bm\Lambda} = {\rm diag}(0, 0, \bm\Lambda)$ and $\tilde{\bm\Lambda}^A$ is anti-symmetric. 
\end{proof}

\section{Conservative free energy of the CG model}
\label{sec_SI:CG_force}
The equilibrium density distribution of the CG model needs to match the marginal density distribution of the CG variables of the full model. Due to the unresolved atomistic degrees of freedom, the conservative CG potential $U(\mb Q) = - \beta^{-1} \ln \left [\int  \delta (\bm\phi^{Q}(\mb q) - \bm Q)  \rho_0(\mb q) \diff \mb q \right] $ (up to a constant) generally encodes the many-body interactions even if the full MD
force field is governed by two-body interactions.  As shown in the previous study \cite{Lei_Cas_2010,hijon2010mori}, accurate modeling of this many-body potential $U(\mb Q)$ is crucial for predicting the static/equilibrium structure properties such as the radial distribution, angle  (i.e., three-body) distribution, and the equation of state. It provides the starting point for the present study focusing on constructing reliable reduced models that accurately predict the non-equilibrium processes on the collective scale.

To establish a fair comparison among the various CG models, we use the \emph{same} conservative CG potential $U(\mb Q)$ constructed by DeePCG \cite{Zhang_DeePCG_JCP_2018} method for all the CG models. As shown in Fig.~\ref{fig:RFD}, all the CG models can accurately recover the radius distribution function (RDF) of the full MD model, where the standard pairwise approximation shows limitations.  This result validates the accuracy of the constructed $U(\mb Q)$. Therefore, the different non-equilibrium properties predicted by the various CG models (presented in the main manuscript) arise from the different formulations of the memory term $\mb K(\mb Q, t)$, which is the main focus of the present study.  

\begin{figure}
    \centering
    \includegraphics[width=0.5\textwidth]{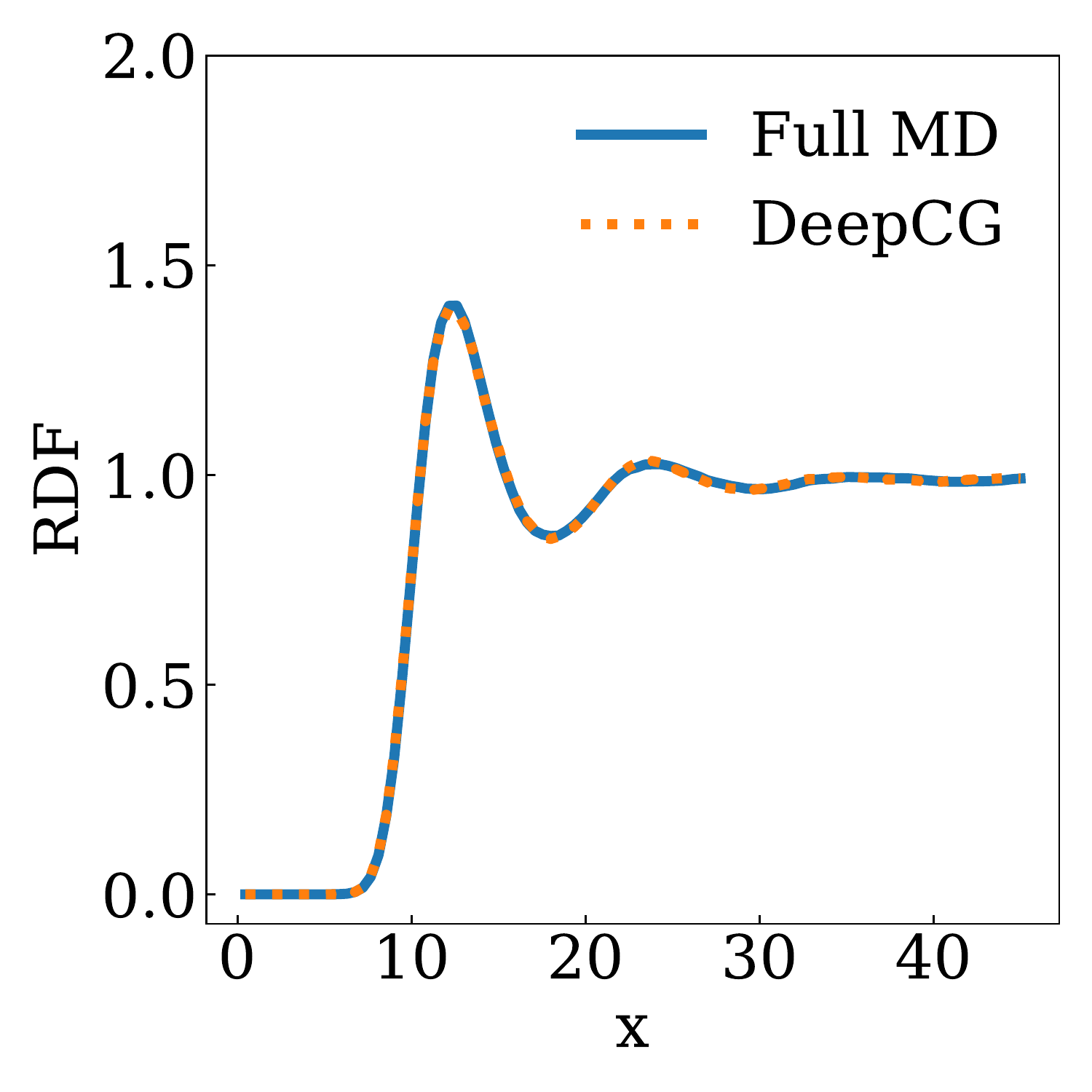}
    \caption{The radius distribution function (RDF) of the full MD and various CG models with the same conservative CG potential $U(\mb Q)$ constructed by the DeepCG model.}
    \label{fig:RFD}
\end{figure}

\section{Symmetry-preserving neural network representation}
\label{sec_SI:NN}

Preserving the physical symmetry constraints is crucial for both the accuracy and the generalization ability of the constructed ML-models. Besides the conservative potential $U(\mb Q)$, the constructed memory term will need to satisfy the translation$\mhyphen$ and permutation$\mhyphen$invariance, as well as the rotation$\mhyphen$symmetries. Let $\mathcal{T}_{\mb b}$, $\mathcal{R}_{\mathcal{U}}$, and $\mathcal{P}_\sigma$ denote the translation, rotation, and permutation operator whose actions on a general function $ \mathcal{F}( \mb Q_1, \cdots ,\mb Q_M )$ defined by
\begin{equation}
\begin{split}
\mathcal{T}_{\mb b} \mathcal{F}( \mb Q_1, \cdots ,\mb Q_M ) &:=  \mathcal{F}( \mb Q_1 +\mb b, \cdots ,\mb Q_M +\mb b), \\
\mathcal{R}_{\mathcal{U}} \mathcal{F}( \mb Q_1, \cdots ,\mb Q_M ) &:= \mathcal{F} ( \mb Q_1 \mathcal{U}, \cdots ,\mb Q_M \mathcal{U}), \\
\mathcal{P}_\sigma \mathcal{F}( \mb Q_1, \cdots ,\mb Q_M ) &:= \mathcal{F} ( \mb Q_{\sigma(1)} , \cdots ,\mb Q_{\sigma(M)}), \\
\end{split}
\end{equation}
where $\mb b\in \mathbb{R}^3$ is a position vector, $\mathcal{U} \in \mathbb{R}^{3\times3}$ is an orthogonal matrix and $\sigma$ is an arbitrary permutation of the set of indices. The components of the constructed memory will need to satisfy the symmetry constraints
\begin{equation}
\begin{split}
\mathcal{T}_{\mb b} \bm\Xi_{ij}(\mb Q_1, \cdots,\mb Q_M ) &= \bm\Xi_{ij}(\mb Q_1, \cdots,\mb Q_M) \\ 
\mathcal{R}_{\mathcal{U}} \bm\Xi_{ij}(\mb Q_1, \cdots,\mb Q_M) &= \mathcal{U} \bm\Xi_{ij}(\mb Q_1, \cdots,\mb Q_M) \mathcal{U}^T \\
\mathcal{P}_\sigma \bm\Xi_{ij}(\mb Q_1, \cdots,\mb Q_M ) &= \bm\Xi_{\sigma(i) \sigma(j)}(\mb Q_{\sigma(1)}, \cdots,\mb Q_{\sigma(M)}),    
\end{split} 
\label{eq_SI:symmetry_condition}
\end{equation}

\begin{proposition} 
The representation $\displaystyle \bm\Xi_{ij} = \sum_{k=1}^K h_k(\hat{\mb Q}_{ij}^T \hat{\mb Q}_{ij} ) \left(\hat{\mb Q}_{ij}^k\right) \left(\hat{\mb Q}_{ij}^k\right)^T + h_0(\hat{\mb Q}_{ij}^T \hat{\mb Q}_{ij} ) \mb I$ preserves the symmetry conditions \eqref{eq_SI:symmetry_condition}, where $\displaystyle \hat{\mb Q}_i^k = \mb Q_i + \sum_{l\in\mathcal{N}_i} f^k(Q_{il}) \mb Q_{il}$ represents the local environment-determined features (generalized coordinate) for the $i\mhyphen$th particle, $\mb f: \mathbb{R} \to \mathbb{R}^K$ and $\mb h: \mathbb{R}^{K\times K}\to \mathbb{R}^{K+1}$ are two encoder functions. 
\end{proposition}

\begin{proof}
We note that $\mathcal{T}_{\mb b}\mb Q_{ij} = \mathcal{T}_{\mb b}\mb Q_{i} - \mathcal{T}_{\mb b}\mb Q_{j}=\mb Q_{ij}$, $\mathcal{T}_{\mb b}Q_{ij} = \left\Vert\mathcal{T}_{\mb b}\mb Q_{i} - \mathcal{T}_{\mb b}\mb Q_{j}\right\Vert = Q_{ij}$, $\mathcal{R}_{\mathcal{U}} \mb Q_{ij} =  \mathcal{U} \mb Q_{ij}$, $\mathcal{R}_{\mathcal{U}} Q_{ij} = Q_{ij}$, $\mathcal{P}_\sigma \mb Q_{ij} = \mb Q_{\sigma(i)\sigma(j)}$, and  $\mathcal{P}_\sigma Q_{ij} = Q_{\sigma(i)\sigma(j)}$.   
Therefore, for arbitrary indices $i$ and $k$, the feature $\hat{\mb Q}_i^k$ satisfy the following symmetry conditions   
\begin{equation}
\begin{split}
\mathcal{T}_{\mb b} \hat{\mb Q}_i^k &= \mathcal{T}_{\mb b}\mb Q_i + \sum_{l\in\mathcal{N}_i} f^k(\mathcal{T}_{\mb b} Q_{il}) \mathcal{T}_{\mb b}\mb Q_{il} = \hat{\mb Q}_i^k + \mb b\\
\mathcal{R}_{\mathcal{U}} \hat{\mb Q}_i^k &= \mathcal{R}_{\mathcal{U}}\mb Q_i + \sum_{l\in\mathcal{N}_i} f^k(\mathcal{R}_{\mathcal{U}} Q_{il}) \mathcal{R}_{\mathcal{U}}\mb Q_{il} = \mathcal{U} \hat{\mb Q}_i^k  \\
\mathcal{P}_\sigma \hat{\mb Q}_{i}^k &= \mathcal{P}_\sigma\mb Q_i + \sum_{l\in{\mathcal{N}_{\sigma(i)}}} f^k(\mathcal{P}_\sigma Q_{il}) \mathcal{P}_\sigma \mb Q_{il} = \hat{\mb Q}_{\sigma(i)}^k,
\end{split}
\label{eq_SI:symmetry_Q}
\end{equation}
where we have used the fact that $\sum_{l} f(r_l)\mb r_l$ is permutational invariant for the last equation.

Therefore, we have $\mathcal{T}_{\mb b}\hat{\mb Q}_{ij} = \mathcal{T}_{\mb b}\hat{\mb Q}_{i} - \mathcal{T}_{\mb b}\hat{\mb Q}_{j}=\hat{\mb Q}_{ij}$, $\mathcal{T}_{\mb b}\hat{Q}_{ij} = \|\mathcal{T}_{\mb b}\hat{\mb Q}_{i} - \mathcal{T}_{\mb b}\hat{\mb Q}_{j}\|= \hat{Q}_{ij}$, $\mathcal{R}_{\mathcal{U}} \hat{\mb Q}_{ij} = \mathcal{U} \hat{\mb Q}_{ij} $, $\mathcal{R}_{\mathcal{U}} \hat{Q}_{ij} = \hat{Q}_{ij}$, $\mathcal{P}_\sigma \hat{\mb Q}_{ij} = \hat{\mb Q}_{\sigma(i)\sigma(j)}$, and  $\mathcal{P}_\sigma \hat{Q}_{ij} = \hat{Q}_{\sigma(i)\sigma(j)}$. Thus, for arbitrary indices $i,j$ and $k$, the encoder functions $h_k(\hat{\mb Q}_{ij} \hat{\mb Q}_{ij}^T) $ satisfy the following symmetry condition
\begin{equation}
\begin{split}
\mathcal{T}_{\mb b} h_k(\hat{\mb Q}_{ij}^T \hat{\mb Q}_{ij}) &=  h_k((\mathcal{T}_{\mb b} \hat{\mb Q}_{ij}){}^T\mathcal{T}_{\mb b} \hat{ \mb Q}_{ij}) = h_k(\hat{\mb Q}_{ij}^T \hat{\mb Q}_{ij})\\
\mathcal{R}_{\mathcal{U}} h_k(\hat{\mb Q}_{ij}^T \hat{\mb Q}_{ij}) &=  h_k((\mathcal{R}_{\mathcal{U}}\hat{\mb Q}_{ij}){}^T \mathcal{R}_{\mathcal{U}}\hat{\mb Q}_{ij} ) =h_k(\hat{\mb Q}_{ij}^T \hat{\mb Q}_{ij}) \\
\mathcal{P}_\sigma h_k(\hat{\mb Q}_{ij}^T \hat{\mb Q}_{ij}) &= h_k(\hat{\mb Q}_{\sigma(i)\sigma(j)}^T \hat{\mb Q}_{\sigma(i)\sigma(j)}).
\end{split}
\label{eq_SI:symmetry_h}
\end{equation}
Plugging Eq. \eqref{eq_SI:symmetry_h} into the definition of $\bm\Xi_{ij}$ yields \eqref{eq_SI:symmetry_condition}.
\end{proof}

\section{Training Details}
\label{sec_SI:training}
With the equilibrium stage presented in Sec. \ref{sec_SI:MD}, we use constrained dynamics to collect samples of the instantaneous force $\mb F(t)$ on individual molecules with a fixed configuration $\bm Z:= [\tilde{\mb Q}, \tilde{\mb P}]$, where $\tilde{\mb Q}$ and $\tilde{\mb P}$ represent the COMs and total momentum of the individual molecules. As they are linear functions of the full phase space vector $\mb z = [\mb q, \mb p]$, the constraint dynamics (see Ref. \cite{hijon2010mori}) for the $j\mhyphen$th atomistic particle associated with the $i\mhyphen$th molecule follows 
  \begin{equation}
  \begin{split}
   \dot{\mb q}_j &= m^{-1} \mb p_j -   \tilde{\mb Q}_i \\
   \dot{\mb p}_j &= -\nabla_{\mb q_j} V(\mb q) + \frac{1}{N_m} \sum_{k\in \mathcal{N}_i} \nabla_{\mb q_k} V(\mb q)
  \end{split}  
  \label{eq_SI:constraint}
  \end{equation}
where $V(\mb q)$ is the potential function of the full MD model and $N_m$ is the number of atoms per molecule.  With $\mb Z(0) = \bm Z$, we have $\mb Z(t) \equiv \bm Z$ for $t > 0$ under \eqref{eq_SI:constraint}.  The memory kernel can be sampled from the time correlation as 
\begin{equation}
\begin{split}
\mb K_{\rm MZ}(\bm Z , t) = \left \langle  \delta \mb F(t) \delta \mb F(0)^T\right\rangle_{\bm Z} ,
\end{split}
\end{equation}
where $\delta \mb F = \mb F - \mathcal{P}_{\bm Z }(\mb F)$ is the fluctuation force on individual molecules and $\mathcal{P}_{\bm Z }(\mb F)$ is the mean force obtained from the many-body potential $U(\mb Q)$ discussed in \ref{sec_SI:CG_force}.  We collect two configuration samples consisting of $486$ molecules. For each configuration, $5000$ independent ensembles are conducted with a production stage of $500000$ steps to compute the correlation function. 

The encoder functions $f$ and $h$ are parameterized as 4-layer fully connected neural networks. Each hidden layer consists of $10$ neurons. The number of state-dependent features is set to be $K=10$ and the number of non-Markovian features $n=5$.

The NNs are trained by Adam \cite{Kingma_Ba_Adam_2015} for 1000000 steps. For each step, 5 targeted CG particles and their neighbors within the cutoff will be selected as one training set.  The initial learning rate is $1\times 10^{-3}$ and the decay rate is $0.5$ per 100000 steps.


%

\end{document}